\documentclass[11pt,a4paper]{article}

\usepackage[left=2.5cm, right=2.5cm, top=2cm, bottom=2cm]{geometry}

\usepackage{authblk}
\usepackage{graphicx} % Required for inserting images
\usepackage{natbib}
\usepackage{amsmath,amssymb}
\usepackage{mathtools, amsthm}
\usepackage{tikz}
\usetikzlibrary{arrows.meta, positioning}
\usepackage{tikz-qtree}
\usepackage{enumitem}
\setlist[itemize]{label=\textbullet}
\usepackage{multirow}
\usepackage{comment}
\usepackage{array}
\usepackage[T1]{fontenc}
\usepackage{textcomp}
\usepackage{pgfplots}
\pgfplotsset{compat=1.18}
\usepgfplotslibrary{fillbetween}
\usepackage{iftex}
\usepackage{subcaption}
\usepackage[ruled,vlined,linesnumbered]{algorithm2e}
\usepackage[dvipsnames]{xcolor}
\definecolor{darkblue}{rgb}{0.0, 0.0, 0.55}
\definecolor{chicago-maroon}{RGB}{128,0,0}
\usepackage[bookmarks, colorlinks=true, plainpages = false, citecolor= darkblue, linkcolor = chicago-maroon, anchorcolor = red, urlcolor = darkblue]{hyperref}

\newtheorem{proposition}{Proposition}
\newtheorem{remark}{Remark}
\newtheorem{example}{Example}

\title{Toward Variation-Independent Regression by Composition \\
\large{Discussion on ``Regression by Composition'' by Farewell, Daniel, Stensrud, and Huitfeldt}}
\author[1]{Ruixuan Zhao \thanks{ruixuan.zhao@utoronto.ca}} 
\author[2]{Oliver Dukes \thanks{oliver.dukes@ugent.be}}
\author[1,3]{Linbo Wang \thanks{linbo.wang@utoronto.ca}}
\author[4]{Lin Liu \thanks{linliu@sjtu.edu.cn}}
\affil[1]{Department of Computer and Mathematical Sciences, University of Toronto, Toronto, Canada}
\affil[2]{Department of Mathematics, Computer Science, and Statistics, Ghent University, Ghent, Belgium}
\affil[3]{Department of Statistical Sciences, University of Toronto, Toronto, Canada}
\affil[4]{Institute of Natural Sciences, Shanghai Jiao Tong University, Shanghai, China}
\date{}

\usepackage{setspace}
\begin{document}
 
\maketitle

We congratulate the authors on their insightful work, which broadens the scope of regression modeling beyond the classical generalized linear model paradigm with a fixed link function.

One structural issue that seems worth further reflection concerns variation dependence. When one starts from a reference distribution and introduces covariate effects through successive flows, the requirement that the resulting object remains a valid probability distribution can induce constraints among the parameters. The Appendix gives a simple example illustrating how such constraints may arise.

This, in turn, raises a natural question: can one construct a framework in which the parameters of interest remain variation independent? A useful starting point is the strategy of \cite{richardson2017modeling}, which begins with the target effect measure and then introduces nuisance components chosen to be variation independent of it. For instance, if the target parameter is the risk ratio
\[
\mathrm{RR}:=\frac{P(Y=1\mid \text{trt}=1)}{P(Y=1\mid \text{trt}=0)},
\]
one may pair it with the odds product of \cite{richardson2017modeling},
\[
\mathrm{OP}:=\frac{P(Y=1\mid \text{trt}=1)P(Y=1\mid \text{trt}=0)}{\{1-P(Y=1\mid \text{trt}=1)\}\{1-P(Y=1\mid \text{trt}=0)\}},
\]
which is variation independent of the risk ratio. This suggests a broader question in the present setting: when covariates act on the outcome through multiple mechanisms, so that interest centers on several effect measures rather than a single one, is it still possible to construct nuisance components that preserve variation independence?

As a concrete example, consider the DAG in Figure \ref{fig1}, with baseline covariate $L_0$, binary treatments $A_0$ and $A_1$, and binary outcome $Y$.

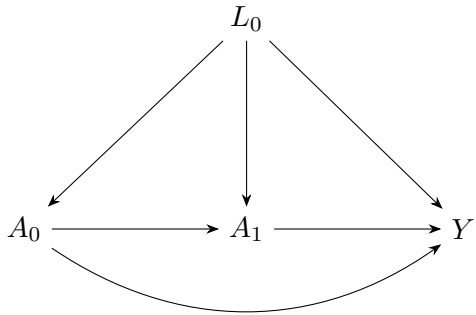
\begin{figure}[htp]
\centering
\begin{tikzpicture}[>=Stealth, node distance=2.2cm]
  \node (A0) {$A_0$};
  \node (A1) [right=of A0] {$A_1$};
  \node (Y)  [right=of A1] {$Y$};
  \node (L0) [above=of A1] {$L_0$};

  \draw[->] (L0) -- (A0);
  \draw[->] (L0) -- (A1);
  \draw[->] (L0) -- (Y);

  \draw[->] (A0) -- (A1);
  \draw[->] (A1) -- (Y);

  \draw[->, bend right=35] (A0) to (Y);
\end{tikzpicture}
\caption{A DAG with baseline covariate $L_0$, sequential binary treatments $(A_0,A_1)$, and binary outcome $Y$.}
\label{fig1}
\end{figure}

Let $p_{a_1,a_0}(l_0):=\Pr(Y=1\mid A_1=a_1, A_0=a_0, L_0=l_0)$. Suppose one is interested in three effect measures: 
\begin{align*}
& \mathrm{RR}_0(l_0):=\frac{p_{0,1}(l_0)}{p_{0,0}(l_0)},\\
& \mathrm{OR}_{1,0}(l_0):=\frac{p_{1,0}(l_0)/\{1-p_{1,0}(l_0)\}}{p_{0,0}(l_0)/\{1-p_{0,0}(l_0)\}},\\
& \mathrm{RR}_{1,1}(l_0):=\frac{p_{1,1}(l_0)}{p_{0,1}(l_0)}.
\end{align*}
In this case, the $l_0$-specific generalized odds product \citep{wang2023coherent}, defined as
\begin{align*}
\mathrm{GOP}(l_0):=\frac{\Pi_{a_1=0,1}\Pi_{a_0=0,1} p_{a_1,a_0}(l_0)}{\Pi_{a_1=0,1}\Pi_{a_0=0,1} \{1-p_{a_1,a_0}(l_0)\}},
\end{align*}
yields a variation independent parameterization.

\begin{proposition}\label{prop1}
For each $l_0$, the map
\begin{align}\label{map1}
\left(\mathrm{RR}_0(l_0), \mathrm{OR}_{1,0}(l_0), \mathrm{RR}_{1,1}(l_0), \mathrm{GOP}(l_0)\right)
\rightarrow \left(p_{a_1,a_0}(l_0), a_1,a_0\in\{0,1\}\right)
\end{align}
is a bijection from $(\mathbb{R}^+)^4$ to $(0,1)^4$. Furthermore, $\mathrm{RR}_0(l_0)$, $\mathrm{OR}_{1,0}(l_0)$, $\mathrm{RR}_{1,1}(l_0)$, and $\mathrm{GOP}(l_0)$ are variation independent.
\end{proposition}

This example suggests that one possible way forward is to begin with a collection of target effect measures and then seek nuisance components that preserve variation independence without imposing additional constraints on the parameters. The same idea may also extend beyond the binary treatment setting to categorical and continuous treatments \citep{yin2022multiplicative}.

\begin{remark}
Once \(\mathrm{RR}_{1,1}(l_0)\) is replaced by the survival ratio \(\mathrm{SR}_{1,1}(l_0):=\frac{1-p_{1,1}(l_0)}{1-p_{0,1}(l_0)}\), the target effect measures are no longer variation independent, because \(\bigl(\mathrm{RR}_0(l_0), \mathrm{SR}_{1,1}(l_0)\bigr)\) is restricted to 
\[
\mathcal{D}
=
\left\{
(r,s)\in(\mathbb{R}^+)^2:\ s(1-r)<1
\right\}.
\]
\end{remark}

\newpage

\section*{Appendix}
\begin{example}(A simple RbC example showing variation dependence)
Consider a simple regression-by-composition (RbC) model $$\mathtt{y=Ber(1/2) ~|~ ScRisk1(1+trt)},$$  
which implies $P(Y=1\mid \text{trt}) = \frac{1}{2}\eta(\text{trt})$ and $\eta(\text{trt}) = \exp{(\alpha+\beta\text{trt})}$ corresponds to the linear predictor for flow $\mathtt{ScRisk1(1+trt)}$. Here, the risk ratio $\frac{P(Y=1\mid \text{trt}=1)}{P(Y=1\mid \text{trt}=0)}=\exp(\beta)$ and the baseline risk $P(Y=1\mid \text{trt}=0)=\frac{1}{2}\exp(\alpha)$ are not variation independent, as the validity of the probability model requires $\eta(\text{trt})\leq 2$. Furthermore, with more flows, the parameters are increasingly tied together, so maintaining variation independence is generally no longer possible.
\end{example}

\begin{proof}[Proof of Proposition \ref{prop1}]
We omit the baseline covariate $L_0$ throughout this proof. To show that \eqref{map1} is a bijection, let $\mathbf{c}=(c_1,c_2,c_3,c_4)$ be a vector in $(\mathbb{R}^+)^4$. It suffices to show that for any $\mathbf{c} \in (\mathbb{R}^+)^4$, there is one and only one $\mathbf{p}=(p_{a_1,a_0}, a_1,a_0\in \{0,1\})$ such that 
\begin{align}\label{eq1}
\left(\frac{p_{0,1}}{p_{0,0}}, \frac{p_{1,0}/\{1-p_{1,0}\}}{p_{0,0}/\{1-p_{0,0}\}},\frac{p_{1,1}}{p_{0,1}}\right) = (c_1,c_2,c_3)
\end{align}
and 
\begin{align}\label{eq2}
\text{GOP}=c_4.
\end{align}
Let $u:=p_{0,0}$, and then \eqref{eq1} implies that 
\begin{align*}
&p_{0,1}=c_1u,\\
& p_{1,0} = \frac{c_2u}{1-u+c_2u},\\
& p_{1,1} = c_1c_3u.
\end{align*}
Thus, once $u\in (0,1)$ is specified, the remaining three probabilities $p_{0,1}$, $p_{1,0}$ and $p_{1,1}$ are uniquely determined. Also, it follows from  $\mathbf{p}\in (0,1)^4$ and $\mathbf{c}\in (\mathbb{R}^+)^4$ that $u\in \left(0,\min\left\{1,\frac{1}{c_1},\frac{1}{c_1c_3}\right\}\right)$. It is easy to see that 
\begin{align*}
\left\{\mathbf{p}\in (0,1)^4: \eqref{eq1}\text{ holds}\right\} = \left\{\psi(u): u\in \left(0,\min\left\{1,\frac{1}{c_1},\frac{1}{c_1c_3}\right\}\right)\right\},
\end{align*}
where $\psi(u)=\left(u,c_1u,\frac{c_2u}{1-u+c_2u},c_1c_3 u\right)$. Thus, for any $\mathbf{p}\in (0,1)^4$ such that \eqref{eq1} holds, the constraint \eqref{eq2} can equivalently be expressed as 
\begin{align*}
\log(\text{GOP}) = \log \left(\frac{u}{1-u}\cdot\frac{c_1u}{1-c_1u}\cdot\frac{c_2u}{1-u}\cdot \frac{c_1c_3u}{1-c_1c_3u}\right)=\log(c_4).
\end{align*}
Let $f(u)=\log \left(\frac{u}{1-u}\cdot\frac{c_1u}{1-c_1u}\cdot\frac{c_2u}{1-u}\cdot \frac{c_1c_3u}{1-c_1c_3u}\right) - \log(c_4)$ and we would like to show that $f(u)$ has one and only one solution in $\left(0,\min\left\{1,\frac{1}{c_1},\frac{1}{c_1c_3}\right\}\right)$. Particularly, we have 
\begin{align*}
\frac{\partial f(u)}{\partial u} &= \frac{\partial}{\partial u}\left\{4\log(u)-2\log(1-u)-\log(1-c_1u)-\log(1-c_1c_3u)\right\} \\
&= \frac{4}{u} + \frac{2}{1-u} + \frac{c_1}{1-c_1u} + \frac{c_1c_3}{1-c_1c_3u}>0,
\end{align*}
implying $f(u)$ is monotone on $\left(0,\min\left\{1,\frac{1}{c_1},\frac{1}{c_1c_3}\right\}\right)$. Furthermore, $\lim_{u\rightarrow 0}f(u)\rightarrow -\infty$ and $\lim_{u\rightarrow \min\left\{1,\frac{1}{c_1},\frac{1}{c_1c_3}\right\}}f(u)\rightarrow +\infty$ and $f(u)$ is continuous in $\left(0,\min\left\{1,\frac{1}{c_1},\frac{1}{c_1c_3}\right\}\right)$. Therefore, $f(u)$ has one and only one solution in $\left(0,\min\left\{1,\frac{1}{c_1},\frac{1}{c_1c_3}\right\}\right)$ so that there is one and only one $\mathbf{p}\in (0,1)^4$ such that both \eqref{eq1} and \eqref{eq2} hold.

\end{proof}

\bibliographystyle{apalike}
\bibliography{main.bib}

\end{document}